\documentclass[10pt,conference,letterpaper]{IEEEtran}

\usepackage[utf8]{inputenc}

\PassOptionsToPackage{table,svgnames,dvipsnames}{xcolor}

\usepackage{balance}
\usepackage{xspace}
\usepackage[utf8]{inputenc}
\usepackage[T1]{fontenc}
\usepackage[sc]{mathpazo}
\usepackage{graphicx}
\usepackage{scrhack} %
\usepackage{listings}
\usepackage{lstautogobble}
\usepackage{tikz}
\usepackage{pgfplots}
\usepackage{pgfplotstable}
\usepackage{booktabs}
\usepackage[hidelinks]{hyperref} %

\usepackage{amsthm}

\newtheorem{lemma}{Lemma}
\newtheorem{theorem}{Theorem}

\usepackage{amsmath}
\usepackage{amssymb}
\usepackage[]{algorithm2e}
\SetKwRepeat{Do}{do}{while}
\usepackage{listings}
\usepackage{subcaption}
\usepackage{thmtools}
\usepackage{thm-restate}
\usepackage{url}

\usepackage[normalem]{ulem}

\usepackage{hyperref}

\usepackage{cleveref}

\usepackage{graphicx,wrapfig}

\usepackage{caption}
\usepackage{float}

\pgfplotsset{compat=newest}
\pgfplotsset{
  cycle list={TUMBlue\\TUMAccentOrange\\TUMAccentGreen\\TUMSecondaryBlue2\\TUMDarkGray\\},
}

\newcommand{\remove}[1]{}

\newcommand{\AStrict}{\textsc{Strict}\xspace}
\newcommand{\AFlex}{\textsc{Flex}\xspace}

\newcommand{\erf}{ERF\xspace}
\newcommand{\nrf}{NRF\xspace}

\newcommand{\aprspstrict}{\textsc{APRSP\_Lorenz}\xspace}
\newcommand{\aprspflex}{\textsc{APRSP\_Goel}\xspace}

\title{It's Good to Relax: Fast Profit Approximation for 
Virtual Networks with Latency Constraints}

\makeatletter
\newcommand{\linebreakand}{%
  \end{@IEEEauthorhalign}
  \hfill\mbox{}\par
  \mbox{}\hfill\begin{@IEEEauthorhalign}
}
\makeatother

\author{\IEEEauthorblockN{Robin M\"unk\IEEEauthorrefmark{1}, Matthias Rost\IEEEauthorrefmark{2}, Stefan Schmid\IEEEauthorrefmark{3}, Harald R\"acke\IEEEauthorrefmark{1}}
	\IEEEauthorblockA{\IEEEauthorrefmark{1}Technical University of Munich}
	\IEEEauthorblockA{\IEEEauthorrefmark{2}SAP SE \& Technische Universit\"at Berlin}
	\IEEEauthorblockA{\IEEEauthorrefmark{3}University of Vienna}}

\begin{document}

\maketitle

\begin{abstract}
This paper proposes a new approximation algorithm for the offline Virtual Network Embedding Problem (VNEP) with latency constraints. Given is a set of virtual networks with computational demands on nodes and bandwidth demands together with latency bounds on the edges. The VNEP's task is to feasibly embed a subset of virtual networks on a shared physical infrastructure, e.g., a data center, while maximizing the attained profit.
In contrast to existing works, our approximation algorithm \textnormal{\AFlex} allows for (slight) violations of the latency constraints in order to greatly lower the runtime. To obtain this result, we use a reduction to the Restricted Shortest Path Problem (RSP) and leverage a classic result by Goel et al.
We complement our formal analysis with an extensive simulation study demonstrating the computational benefits of our approach
empirically.
Notably, our results generalize to any other 
additive edge metric besides latency,
including loss probability. 
\end{abstract}

\section{Introduction}

The Virtual Network Embedding Problem (VNEP) is a fundamental resource allocation
problem in networks and has received significant interest in the network algorithms community
over the last decade. 
The VNEP is motivated by the resource allocation flexibilities available in 
virtualized environments, such as the Cloud, where node and network resources
can be shared and provisioned on demand.
A virtual network provides the illusion of a dedicated network to the user,
although it is realized over a shared infrastructure. To this end,
a virtual network provides resource guarantees both on the nodes (e.g., CPU)
as well as on the edges (e.g., bandwidth).
 
The optimization problem underlying the VNEP is the following. 
We are given a set of request graphs (the virtual networks, sometimes also called ``guest graphs'') 
and a single substrate network (the physical infrastructure, also called the ``host graph'').
For every request graph the task is to either find a feasible embedding that maps each request node to a substrate node 
and every request edge to a path in the substrate graph, or to reject the request. The cumulative resource consumption of the embeddings may then not violate the substrate capacities on both nodes and edges. In this paper we consider unsplittable paths with latency constraints: if a request is admitted, its edges are embedded as simple paths. Every admitted and feasibly embedded request yields a given profit and the goal is to maximize the total profit.

The VNEP is hard so solve in many variants. Even when neglecting the cumulative feasibility constraints, which is known as the Valid Mapping Problem (VMP)~\cite{rost2019parametrized}, the problem remains $\mathcal{NP}$-hard~\cite{ton20hard}. Importantly, solving the VMP is an essential building block for approximation algorithms for the VNEP~\cite{rost2019parametrized, rostTON20randround}.

In this paper we consider latency as an example of secondary edge demand besides bandwidth. 
Latency has become a critical metric for many applications, e.g., in the context of industrial or tactile networks.
In addition, the applications expected to emerge around 5G require very
low latency, deterministic packet delivery and high availability~\cite{jiang2018low,li20175g}.
Our results can however be applied to any additive edge metric, like hop count and, interestingly, even packet loss probability, which is by itself multiplicative but becomes additive when taking its logarithm. 

\subsection{Contributions}

This paper presents a novel, fast and practical approximation algorithm \AFlex
for the VNEP with latency constraints (or any other additive constraint on edge metrics).
\AFlex provides both analytical approximation guarantees and performs well in practice, as demonstrated in our computational evaluation. 
\AFlex is based on the insight that a slight relaxation of the latency
guarantees can result in significantly faster and hence more practical
solutions.
The latency violations can be made arbitrarily small, by trading off 
for a longer runtime.
To achieve this, 
\AFlex builds upon the dynamic programming and randomized
rounding framework by Rost et al.~\cite{rost2019parametrized},
which solves an all-pairs Restricted Shortest Path Problem (RSP) as a subroutine.
In order to solve the RSP, 
we employ a classic result by Goel et al.~\cite{goel2001efficient}
which allows, in one execution, to calculate the routes for all
destination nodes at once for a given source node. 

Compared to the state-of-the-art algorithm, which we refer to by \AStrict
as it provides strict latency guarantees (relying on an approximation scheme for the RSP
by Lorenz and Raz~\cite{lorenz2001simple}),
\AFlex is orders of magnitudes faster, sometimes 
reducing the runtime from over nine hours to below three minutes. 
At the same time, the profit approximation and average latency achieved
by \AFlex  is similar to the one obtained by \AStrict. 
We believe that this makes \AFlex a much more practical solution.

To ensure reproducibility, we have made the source code
of our algorithms and our experiments publicly available
\footnote{at \href{http://www.github.com/vnep-approx-latency}{github.com/vnep-approx-latency}}.

\subsection{Related Work}

Our paper builds upon the algorithmic framework by Rost et al.~\cite{rost2019parametrized} 
which uses a dynamic program to jointly optimize the mapping of the request's nodes and
edges, and relies on randomized rounding.
This allows us to reduce part of our task to solving the Restricted Shortest Paths (RSP) problem,
a special case of the Multi-Constrained Optimal Path Problem (MCOP) 
where the goal is to find a source-target path minimizing the costs while respecting 
$K \geq 1$ additive quality-of-service parameters;
the RSP is a MCOP with $K = 1$. 
In the following, we first review the literature on the RSP, and then discuss the VNEP.

\subsubsection{The RSP}
\label{sec:related-work:rsp}
The RSP was shown to be NP-hard in \cite{gary1979computers}.
The first fully polynomial time-approximation scheme (FPTAS) for the  RSP for general graphs was developed by Hassin~\cite{hassin1992approximation} 
in 1992. 
Hassin's algorithm  
gives a $(1+\epsilon)$-approximate cost-minimal path, 
guarantees to respect the given latency bound $T$,
and runs in time $\mathcal{O}(mn(n/\epsilon) \, \log(n/\epsilon))$, where $n$ is the number of nodes and $m$ the number of edges.
This result was then improved in 2001 by Lorenz and Raz~\cite{lorenz2001simple},
to $\mathcal{O}(mn(\log{\log{n}} + 1/\epsilon))$. The authors first compute an upper and lower bound for the costs of the optimal solution in $\mathcal{O}(mn \, \log \, \log \, n)$ time, which they then refine to a $(1+\epsilon)$-approximate solution in $\mathcal{O}(mn / \epsilon)$ time.
Both approaches, the one by Hassin and the one by Lorenz and Raz,
have in common that they re-scale and then discretize the edge-costs based on 
bounds of the optimal solution.
The current fastest FPTAS for the RSP in general graphs runs in time $\mathcal{O}(mn(1/\epsilon + \log{\log{\log{n}}}))$. The first such algorithm was proposed by Xue et al.~\cite{xue2008polynomial} in 2008. 
Under additional assumptions, there exist faster approximation schemes for the RSP. For example the author of~\cite{bernstein2012nearlinear} proposes an algorithm that has near-linear time complexity, but it is non-deterministic and only works in undirected graphs. If the given graph is planar and acyclic (i.e., a DAG) or only has integer edge-costs, the algorithms proposed in~\cite{holzmuller2017improved} runs in $\mathcal{O}(mn/\epsilon)$ time. 

The above algorithms all approximate the cost while safeguarding that the additive edge constraint, e.g., the latency bound, is met. In contrast, Goel et al. proposed in~\cite{goel2001efficient} a different scheme which relaxes the edge constraints while always achieving the optimal cost (or better):
\begin{theorem}[Goel et al.~\cite{goel2001efficient}]\label{theorem:goel}
	For a given graph $G=(V,E)$ and source node $s \in V$ there exists an algorithm that computes paths $P(t)$ from $s$ to each target node $t \in G$ in time \mbox{$\mathcal{O} \Big( (m + n \, \log{n} ) \, D / \epsilon  \Big)$} with $D \le n$  such that for each path $P(t)$ the latency bound is violated by at most a factor of $1+\varepsilon$ while the cost of $P(t)$ lies beneath the optimal cost of any latency bound $s-t$ path.
\end{theorem}
Besides changing the approximation objective, the algorithm by Goel et al. is of particular interest to us as it simultaneously computes paths to all target nodes. This fact allows us to speed up our VNEP approximation.

\subsubsection{The VNEP}
\label{sec:related-work:vnep}
The Virtual Network Embedding Problem has received 
much attention by the networking community over the last decade,
and we refer the reader to the extensive surveys on this topic~\cite{vnep-survey}. 
Much existing work revolves around heuristics~\cite{vnep, Lischka, meng2010improving} 
and exact algorithms based on mixed-integer programming~\cite{vnep} %
which is motivated by the fact that the VNEP is NP-hard and inapproximable~\cite{ton20hard}.
Notwithstanding, there are first results on 
polynomial-time approximation algorithms~\cite{ifip20vnep, rostTON20randround} in the resource augmentation model,
as well as on parametrized exact and approximation algorithms for very restricted problem instances~\cite{rost2019parametrized}. 
To the best of our knowledge, the work~\cite{rost2019parametrized} is the first and only to provide
an approximation algorithm which also accounts for latencies. We refer to the algorithm presented in~\cite{rost2019parametrized} as \AStrict. Comparing with our novel \AFlex algorithm, we show that \AFlex is significantly faster while only introducing negligible latency violations.

We conclude by noting that 
the Virtual Network Embedding Problem is related to various classic 
graph-theoretical problems such
as \emph{VLSI Graph Layout}~\cite{bhatt1984framework}, 
\emph{Graph Labeling}~\cite{chung1988labelings}, and \emph{Subgraph Isomorphism}~\cite{eppstein2002subgraph}.
The \emph{VLSI Graph Layout} problem typically deals with the question of how to 
minimize the layout area of a circuit on a chip, which corresponds to embedding
a request graph onto a two-dimensional grid such that the \emph{embedding area}
(the product of vertical and horizontal lines) is minimized~\cite{bhatt1984framework}.
In the \emph{Graph Labeling} problem~\cite{chung1988labelings}, the nodes of a (substrate) graph $G$
need to be labelled by distinct nodes of a (request) graph $H$ while embedding the edges of $H$ onto $G$ with a popular objective being to minimize
the total sum of distances. 
\emph{Graph Labeling} on line substrates is widely known as \emph{Minimum Linear Arrangement}~\cite{diaz2002survey}.
In contrast to the above problems, the VNEP 
explicitly allows for mapping several request nodes to a single substrate node and introduces capacities on nodes and edges, rendering it significantly harder to solve.

\subsection{Organization}

The remainder of this
paper is organized as follows.
In Section~\ref{section:problemdefinition},
we introduce the model and preliminaries. 
We then present our algorithms in 
Section~\ref{section:approach}.
Section~\ref{section:evaluation} discusses
implementation details and reports on our evaluation
results.
We conclude in Section~\ref{section:conclusion}.

\section{Model and Preliminaries}\label{section:problemdefinition}

The \textbf{substrate network} is given as a directed graph  
$G_S = (V_S, E_S)$. Each component of the network, 
that is, each substrate node $v_S \in V_S$ and each substrate edge $e_S \in E_S$, 
has a capacity $d_S : G_S \rightarrow \mathbb{R}_{\geq0}$. 
For nodes the capacity may refer, e.g., to the number of available CPU cores, and restricts the number of virtual nodes that can be mapped onto it. 
Further, each substrate component $x \in G_S$ may be attributed with a cost value 
$c_S(x) \in \mathbb{R}_{\geq0}$ for its usage. Edge latencies are given 
by the function $l_S : E_S \rightarrow \mathbb{R}_{\geq0}$ and represent 
the time delay between two neighboring substrate nodes.

A \textbf{request} is likewise represented by a directed graph $G_r = (V_r, E_r)$ with demands $d_r : G_r \rightarrow \mathbb{R}_{\geq0}$ for each virtual component and an associated latency bound \mbox{$T_r \in \mathbb{R}_{\geq0}$} such that all virtual edges of $E_r$ must be embedded with a lesser or equal latency. Every request $r$ yields a given profit $b_r \in \mathbb{R}_{\geq0}$ if it is successfully embedded in the substrate. We denote by $d_{\text{max}}(r, x)$  the maximal demand of any request element on the substrate resource $x \in G_S$.

A \textbf{mapping} represents how a request is embedded in the substrate. In our model we allow the specification of a set of forbidden nodes and edges with each request, i.e., the virtual nodes and edges may only be mapped on a subset of substrate nodes and edges. Formally, a \emph{valid} mapping of request $r$ onto the substrate $G_S$ is defined as a tuple $m_r=(m_r^V, m_r^E)$ of functions, such that:
 \begin{itemize}
    \item The function $m_r^V : V_r \rightarrow V_S$ assigns a \emph{valid} substrate node to every virtual node. A substrate node is valid for a request node if it has sufficient capacity and if it is not in the set of forbidden nodes for $r$.
    
    \item The function $m_r^E : E_r \rightarrow \mathcal{P}_S$ maps each virtual edge $(i,j) \in E_r$ to a \emph{valid} simple path in the substrate network connecting $m^V_r(i)$ to $m^V_r(j)$.
 \end{itemize}
 
\noindent With regard to latencies, a mapping is further called \emph{valid}
 \begin{itemize}
     \item[-] \emph{under the strict latency constraint} if it additionally fulfills
         $ \sum_{(u,v) \in m_r^E(i,j)} l_S(u,v) \leq T_r $ for all $(i,j) \in E_r $
        such that all latency bounds are met exactly, or 
     \item[-] \emph{under  a $(1+\epsilon)$-approximate latency constraint} for some $\epsilon > 0$ if it is valid and fulfills
        \begin{equation}\label{def:approxLatConstr}
            \sum \nolimits_{(u,v) \in m_r^E(i,j)} l_S(u,v) \leq (1+\epsilon) \cdot T_r
        \end{equation}
    for $(i,j) \in E_r$, allowing for small latency violations.
\end{itemize} 

For a valid mapping $m_r=(m_r^V, m_r^E)$ the induced resource allocation on a substrate element is denoted by $A(m_r, v) = \sum_{i \in V_r:m_r^V(i)=v} d_r(i)$ for nodes $v \in V_S$ and $A(m_r, e) = \sum_{(i,j) \in E_r:e \in m_r^E(i, j)} d_r(i, j)$ for edges $e \in E_S$. Furthermore, we denote by $A_{\text{max}}(r,x)$ the maximum allocation on $x \in G_S$ among all valid mappings. For a single request $r$ the \textbf{Valid Mapping Problem} (VMP) asks to find a valid mapping $m_r$ that minimizes the cost $c(m_r) = \sum_{x \in G_S} c_S(x) \cdot A(m_r, x)$.

For the definition of the VNEP a set of requests $\mathcal{R}$ is given.
We refer to a set of mappings $\{ m_r \} _{r \in \mathcal{R}' }$ for a subset of requests $\mathcal{R}' \subseteq \mathcal{R}$ as a feasible embedding
iff. the cumulative 
resource allocation on any substrate element does not exceed its capacity, i.e., if for all $x \in G_S$ it holds $\sum_{r \in \mathcal{R}'} A(m_r, x) \leq d_S(x)$. It is important to note that the validity of mappings only considers the feasibility of single node and edge mappings while the feasibility of embedding takes the cumulative resource allocations of a \emph{set} of mappings into account.
The (offline) \textbf{Virtual Network Embedding Problem} (VNEP) then is to find a feasible embedding $\{ m_r \}_{r \in \mathcal{R}'}$ of a subset of given requests $\mathcal{R}' \subseteq \mathcal{R}$ which maximizes the profit $\sum_{r \in \mathcal{R}'} b_r$.

The VNEP has been shown to be $\mathcal{NP}$-hard and inapproximable in many variants~\cite{ton20hard}. In this paper, we hence consider parametrized approximation algorithms under model relaxations, i.e., algorithms of polynomial runtime for specific graph classes whose solutions allow for capacity and latency violations. As derived by Rost et al.~in~\cite{rost2019parametrized}, to approximate the VNEP it suffices to solve the VMP (exactly) for each request which in turn requires to compute restricted shortest paths under latency constraints. That is, as a subroutine to solve the VMP, the RSP needs to be solved.

In the \textbf{Restricted Shortest Paths Problem} (RSP) we are given a directed graph $G = (V,E)$ where each edge $e \in E$ is associated with a cost $c_e$ and a latency $l_e$, both non-negative. Then for a given source $s \in V$ and target $t \in V$ the goal is to find a cost-minimal path from $s$ to $t$ such that the latency along this path does not exceed a given limit $T \in \mathbb{R}_{\geq0}$.
For the purposes of this paper, the graph $G$ will be the substrate network of the VNEP and the upper limit $T$ will equal the latency bound of the respective request.

Formally, the RSP can be expressed as a constrained optimization problem, which we define below. In the analysis we denote by $C_G(p) := \sum_{e \in p} c_e$ the total costs and by $L_G(p) := \sum_{e \in p} l_e$ the total latencies of a path $p$ in the graph $G$. The value $C_{\text{opt}}(G)$ will represent the costs of a cost-minimal path that satisfies the latency constraint in the graph $G$. Here, $P_{s, t}$ denotes the set of all paths from $s$ to $t$.
\[
\begin{aligned}
    & \underset{p}{\text{minimize}} 
    & & C_G(p) \\
    & \text{subject to} 
    & & L_G(p) \leq T \\
    & & & p \in P_{s, t}
\end{aligned}
\]

Note that the objective function only depends on the cost value of the path $p$ and not the latency value. The latency value is used merely as a constraint. As a result, two \emph{feasible} paths with equal cost value will also have the same objective value, regardless of their latency value.
\remove{
An algorithm $\mathcal{A}$ is an \emph{approximation scheme} for the RSP if for every Input $(I, \epsilon)$, consisting of an instance of the RSP $I$ and an error parameter $\epsilon > 0$, it outputs a solution $s$ such that
    $f_{RSP}(I, s) \leq (1 + \epsilon) \, \cdot C_{\text{opt}}(I)$
where  $f_{RSP}(I, s)$ is the objective value, i. e. cost of the solution-path produced by $\mathcal{A}$. We further introduce the definition of a $(\alpha, \beta)$-algorithm for the RSP. 
}

\section{Algorithms and Analysis}\label{section:approach}

\subsection{Algorithmic Framework}
In order to approximate the VNEP with latency constraints, we build upon  the framework by Rost et al.~\cite{rost2019parametrized}.
Their approach tackles the problem in multiple steps and is parametrized by the treewidth of the request graphs, a measure of similarity to trees, i.e., the algorithm's runtime is only polynomial if the maximal treewidth of the request graphs is a constant. 

The algorithm in~\cite{rost2019parametrized} works as follows. First, for each request graph $G_r$ a tree decomposition $\mathcal{T}_r$ of limited treewidth is computed. It is then shown that the Valid Mapping Problem (VMP) can be solved on this tree representation using the \textsc{DynVMP} algorithm using dynamic programming (in time and space exponential in the request's treewidth). Given the ability to solve the VMP (without latencies), the \emph{fractional} VNEP is then shown to be solvable via column generation techniques where the \textsc{DynVMP} algorithm is used as a separation oracle. This fractional solution can be interpreted as a `probability distribution' over the valid mappings constructed in the column generation step and can be easily converted into a solution to the VNEP via (repeated) randomized rounding. Altogether this approach results in an algorithm that produces approximate solutions to the VNEP \emph{without latencies} in parametrized time:

\begin{theorem}[Rost et al.~\cite{rost2019parametrized}] \label{theorem:approxFactors}
    There exists an approximation for the VNEP without latency constraints, which achieves at least an $\alpha = 1/3$ fraction of the optimal profit and the allocations on nodes and edges are within factors $\beta$ and $\gamma$ of the original capacities respectively with high probability.
    The values of $\beta, \gamma \geq 0$ are defined as 
    $\beta := 1 + \sigma \cdot \sqrt{2 \cdot \Delta(V_S) \cdot \log(\vert V_S \vert)}$ and 
    $\gamma := 1 + \sigma \cdot \sqrt{2 \cdot \Delta(E_S) \cdot \log(\vert E_S \vert)}$ with 
    $\Delta(X) := \max_{x \in X}{\sum_{r \in \mathcal{R} : d_{\text{max}}(r, x) > 0} (A_{\text{max}}(r,x) / d_{\text{max}}(r,x) )^2}$
    being the maximal sum of squared maximal allocation-to-capacity ratios over the resource set $X$ and the maximum demand-to-capacity ratio
    $\sigma :=  \max_{r \in \mathcal{R}, x \in G_S}{d_{\text{max}} (r, x) / d_S(x)}$. The algorithm's runtime is polynomial when the maximal request treewidth is a constant.
\end{theorem}

The paper already outlines how latency constraints can be taken into account within this framework. As latencies only change the notion of validity of mappings and pertain to individual request graphs, the \textsc{DynVMP} algorithm needs to be adapted to return solutions respecting latency constraints. This restriction is handled when edge mappings are calculated by approximating the underlying Restricted Shortest Paths problem for each pair of substrate nodes and each request edge.

As the RSP needs to be solved for every pair of substrate nodes,
the \emph{All-Pairs Restricted Shortest Path Problem} (APRSP) needs to be solved (for each request edge). Given the $\mathcal{NP}$-hardness of the RSP, the APRSP can only be approximated.

\noindent \textbf{Algorithm STRICT.}
As discussed in the related work (cf.~Section~\ref{sec:related-work:rsp}), Lorenz and Raz proposed a strongly polynomial FPTAS to the RSP which can be easily extended to an algorithm for the APRSP by solving  $\mathcal{O}(|V|^2)$ problem instances for all substrate node pairs. We denote this adaption as \aprspstrict.
As the algorithm by Lorenz and Raz finds a path that is guaranteed to meet the latency bound $T$ and whose costs are at most $(1+\epsilon) \cdot C_{\text{opt}}$ if such a path exists, the resulting \textsc{DynVMP} adaption approximates the cost of the optimal valid mapping while strictly respecting all latency constraints.

\noindent \textbf{Algorithm FLEX.}
We propose a new algorithm to solve the APRSP subproblem, called \AFlex, which results from using the procedure by Goel et al.~\cite{goel2001efficient} to calculate latency-constrained shortest paths for the \textsc{DynVMP} algorithm instead of the FPTAS by Lorenz and Raz.

The approach comes with a trade-off. The algorithm by Goel et al. calculates cost-optimal paths at the expense of allowing for a violation of the latency constraint by a factor of up to $(1 + \epsilon)$. The approach starts with a coarse scaling of the edge latencies to integers. The modified problem is solved exactly using dynamic programming resulting in cost-minimal paths for a weakened latency constraint. If all paths are also valid for the $(1+\epsilon)$-approximate constraint, the algorithm terminates. Otherwise the process is repeated with a finer scaling until a solution is found.

The crucial advantage of the procedure by Goel et al. are that one execution gives the results for \emph{all} destination nodes at once for a given start node (cf. Theorem~\ref{theorem:goel}). This leads to a significant decrease in runtime as it only has to be executed $|V_S|$ times to produce paths between all pairs of source and target nodes. This subroutine, \aprspflex, only requires $\vert V_S \vert$ calls to Goel et al.'s  algorithm to prepare the cost and path tables for \textsc{DynVMP}.

Besides the fewer required subroutine calls, the algorithm has some additional benefit. Specifically, the algorithm progressively improves the approximation's quality which allows for early stopping of the algorithm in the case of strict computation time limits. In that case, if a path has been computed, it is cost-optimal for some weaker latency constraint.

\subsection{Consequences for Approximating the VNEP}\label{section:application}

In this section we analyze how the modifications within the \AStrict and the \AFlex algorithm influence the runtime and approximation quality of the surrounding Virtual Network Embedding framework.

Let $n := \vert V_S \vert$ be the number of nodes in the substrate and $m := \vert E_S \vert$ be the number of substrate edges. Then the runtime of a single execution of the FPTAS by Lorenz and Raz is bounded by $\mathcal{O} ( m n  \,  ( \log {\log{n}} + 1 / \epsilon ))$~\cite{lorenz2001simple}, which leads to a runtime for \aprspstrict of $ \mathcal{O} \big ( m n^3  \, ( \log {\log{n}} + 1 / \epsilon) \big ) $.

The path costs are $(1+\epsilon)$-approximated. How this factor translates to the approximation of the embedding profits has been thoroughly analyzed by Rost et al., leading to the following Theorem~\ref{theorem:strict}, which contains their results about the \AStrict algorithm. 

\begin{theorem}[\AStrict, Rost et al.~\cite{rost2019parametrized}]\label{theorem:strict}
    For $n \geq 3$ the \AStrict algorithm finds a solution to the VNEP under latency constraints with a profit of at least $1/3 \cdot (1 + \epsilon) ^{-3/2}$ of the optimal profit with high probability. The resource allocation approximation factors $\beta$ and $\gamma$ are the same as defined in Theorem~\ref{theorem:approxFactors}. The runtime is bounded by $\mathcal{O} ( \text{poly} ( \tau_{Strict}) )$ with
        $$ \tau_{Strict} = 
            \sum_{r \in \mathcal{R}} 
                n^2 \cdot \Big ( \,
                    \vert V_r \vert ^3 \cdot n^{2 \cdot \text{tw}(\mathcal{T}_r)}
                        + m \cdot n \cdot \big( \log {\log{n}} + \frac{1}{\epsilon} \big)
                \Big ), $$ where $tw(\mathcal{T}_r)$ denotes the (minimal) treewidth of the request graph $G_r$ (cf.~\cite{bodlaender1998partial}).
\end{theorem}

For the \AFlex algorithm, the algorithm's runtime is again very much determined by the runtime to solve the APRSP. As analyzed by Goel et al.~\cite{goel2001efficient}, it takes $\mathcal{O} ( (m + n \, \log{n} ) \, n / \epsilon )$ time to calculate the paths from one source node to all target nodes (cf. Theorem~\ref{theorem:goel}). The algorithm is called $n$ times leading to a runtime of $ \mathcal{O} \big ( ( m + n \log{n}  ) \cdot n^2/\epsilon \big ) $ for \aprspflex.

The paths, and therefore the embeddings, are cost-optimal for the (1+$\epsilon$)-relaxed latency constraint. This result carries over similarly to the analysis for \AStrict, with cost-optimal paths, leading to a total profit approximation factor of $\alpha_{Flex} = 1/3$. This result is summarized in the following theorem. 

\begin{theorem}[\AFlex]\label{theorem:contribution:flex}
    For $n \geq 3$ the \AFlex algorithm finds a solution to the VNEP under $(1 + \epsilon)$-approximate latency constraints with a profit of at least $1/3$ of the optimal profit with high probability. The resource allocation approximation factors $\beta$ and $\gamma$ are the same as defined in Theorem~\ref{theorem:approxFactors}. The runtime is bounded by $\mathcal{O} ( \text{poly} ( \tau_{Flex}) )$ with
        $$\tau_{Flex} =  
            \sum_{r \in \mathcal{R}} 
                n^2 \cdot \Big ( \,
                    \vert V_r \vert ^3 \cdot n^{2 \cdot \text{tw}(\mathcal{T}_r)}
                        + \frac{ m + n \log{n}}{\epsilon}
                \Big ). $$ 
\end{theorem}

To substantiate the claim of the above theorem, we argue for its correctness in the following.

\begin{lemma}\label{lemma:flexMappingsValid}
    Any mapping returned by the \textsc{DynVMP} using \aprspflex to compute restricted shortest paths is valid under a  $(1+\epsilon)$-approximate latency bound.
\end{lemma}
\begin{proof}
    Let $m_r=(m_r^V, m_r^E)$ be the mapping returned by the \textsc{DynVMP} procedure using \aprspflex. The validity of the node mapping $m_r^V$ follows from the correctness of the \textsc{DynVMP} procedure without latency considerations. The request edge mapping $m_r^E$ is valid as it maps to the set of paths calculated by Goels algorithm. By Theorem~\ref{theorem:goel} all paths in this set satisfy equation~\ref{def:approxLatConstr}. 
\end{proof}

Next we deduce that the modified \textsc{DynVMP} algorithm functions correctly under the relaxed latency constraint.

\begin{lemma}\label{theorem:contribution.goel}
    The modified \textsc{DynVMP} procedure which uses \aprspflex to calculate restricted shortest paths produces a mapping that is valid under the $(1+\epsilon)$-approximate latency bound and of optimal objective w.r.t. the original latency bound, if such a mapping exists.
\end{lemma}
\begin{proof}
    From Theorem~\ref{theorem:goel} it follows that \aprspflex always returns paths of optimal cost that violate the latency bound by at most a factor of $1+\epsilon$ (cf. Theorem~\ref{theorem:goel}). Whenever a path of latency value at most $T_r$ exists, the algorithm must return a path with latency value bounded by $(1+\epsilon) \cdot T_r$ and of objective at most the optimal cost (cf. Theorem~\ref{theorem:goel}). Hence, if a valid mapping for the original latency constraint exists, there will exist a valid edge mapping under the weakened latency constraint. Therefore a mapping will be produced and by Lemma~\ref{lemma:flexMappingsValid} it will also be valid under the $(1+\epsilon)$-approximate latency bound.

    The \textsc{DynVMP} algorithm accordingly correctly determines whether a valid mapping exists and if so, returns a cost-optimal one as only the path-computation was adapted. Since Goel et al.'s algorithm returns paths of optimal costs with respect to all valid paths that satisfy the strict latency constraint, the constructed mappings will be of optimal objective as the \textsc{DynVMP} algorithm computes optimal node mapping costs and hence optimal overall costs.
\end{proof}

Finally we conclude that the rest of the proof proceeds analogously to the original proof by Rost et al. ~\cite{rost2019parametrized}, namely the modified \textsc{DynVMP} procedure can be used to separate the constraints and serve as separation oracle for the LP that solves the fractional VNEP. This fractional solution is then transformed into valid embeddings using randomized rounding.

To assess the overhead of considering latencies compared to the baseline implementation without latencies, we state the runtime when using  Dijkstra's algorithm for every source node (cf.~Theorem~\ref{theorem:approxFactors}):

{
	\vspace{-8pt}
	\small
\begin{align}
\mathcal{O} 
\left ( 
	\text{poly} \left ( 
		\sum_{r \in \mathcal{R}} n \cdot \left ( 
			\, \vert V_r \vert ^3 \cdot n^{2 \cdot \text{tw}(\mathcal{T}_r) + 1} + ( m + n^2 ) 
			 \right ) 
	\right )  
\right ) . 
\end{align}
}

\section{Empirical Evaluation}\label{section:evaluation}

To complement our theoretical contribution and investigate the performance
of our algorithms in practice, we implemented both approximation
algorithms, \AFlex and \AStrict and evaluated them in realistic settings. Given the limited scalability of \AStrict, our main evaluation only uses small to medium sized substrate networks. To further substantiate the benefits of our novel algorithm \AFlex, we also conduct explorative experiments on larger substrate networks, comparing \AFlex to the baseline algorithm which does not consider latencies.

\subsection{Implementation}\label{section:implementation}

We implemented the two approximation algorithms
in Python 3, building upon the implementation by Rost et al.
\footnote{see \href{https://github.com/vnep-approx}{github.com/vnep-approx}} 
Our implementations of the \AFlex and the \AStrict together with the evaluation are publicly available at \href{http://www.github.com/vnep-approx-latency}{github.com/vnep-approx-latency}.

The implementation of the \AStrict algorithm, which uses the algorithm by Lorenz and Raz~\cite{lorenz2001simple},
closely follows the pseudo code provided in their paper.
The implementation of the \AFlex algorithm, based on the work by Goel et al.~\cite{goel2001efficient},
relies on a dynamic programming subroutine which
assumes integer edge latencies and iterates until some delay threshold is met. 
Because some substrate nodes may not be reachable under latency constraints 
for some source node, costs and paths need not always exist. To reduce memory usage, the implementation only stores costs and paths when they exist.

The runtime of both RSP algorithms was improved using the following optimizations.

 \paragraph{Removing infeasible nodes}
		 For an infeasible target node, the algorithm by Lorenz and Raz can only
		conclude that no valid path exists at the very end of its execution.
        A much faster solution is hence to run one execution of 
		a shortest paths algorithm, using the edge latencies as minimization objective, starting from each node before calling the RSP algorithm. Every node whose distance in latencies is greater than the limit cannot be reached by any feasible path. 
		Conversely, if a node's latency-distance from the source is within the limit, then there must be a feasible path and the algorithm has to find a solution. 
		The experiments showed up to a 20$\times$ faster execution time with this optimization.
	
 \paragraph{Optimizing data structures}
        At numerous points in the execution, the algorithm 
				has to check if a given edge is valid. To speed up 
				ckecking list membership, we store such information 
				in a hash map.
				
 \paragraph{Avoiding re-allocations}
        The algorithms require large tables in which to store their results. 
				It can be seen that the distance table of the 
				algorithm by Lorenz and Raz 
				is of size $\mathcal{O}(n^2/\epsilon)$
				which can be quite large for small $\epsilon$. 
        All tables are reset when necessary, 
				yet the distances tables for both algorithms grow dynamically,
				and in both algorithms the size of the distances table is not constant 
				in the second dimensions between calls to the procedures. 
				We hence initialize the algorithms with some value for the second dimension,
				and only if this value is too small in some iteration, the tables are re-allocated. 
        
Besides our two algorithms for the VNEP with latencies, we also evaluate a baseline
algorithm, henceforth simply called \emph{no latencies} or \emph{baseline}: the current state-of-the-art algorithm for the VNEP \emph{without} latencies~\cite{rost2019parametrized}.

\subsection{Computational Setup}

In the following, we describe the computational setup for our main experiments to compare the performance of \AFlex, \AStrict, and the baseline. For the second set of experiments on larger substrate graphs, the computational setup is given separately in Section~\ref{sec:evaluation:aflex-performance}.

We consider five real-world networks from the Topology Zoo in our evaluation~(see Table \ref{tab:zoo}). To impose meaningful latency limits, we compute the average substrate edge latency $\phi(G_S)$ based on the geographic information of the adjacent nodes stored in the Topology Zoo. 

\begin{table}[t!]
    \centering
	\begin{tabular}{ c|c|c } 
	Substrate Network & Nodes & Edges \\ %
	\hline
	Netrail  & 6 & 20 \\  %
	Eunetworks  & 14 & 38 \\ %
	Noel  & 18 & 50 \\  %
	Oxford  & 19 & 52 \\  %
	Funet  & 25 & 62 %
\end{tabular}
\caption{Networks for comparing \AFlex and \AStrict}
\label{tab:zoo}
\end{table}

The general experiment design closely follows Rost and Schmid and we shortly summarize the key points. Specifically, we employ the same procedures to create substrate and request graphs. The request graph topologies are cactus graphs which are created at random such that each graph has between 4 to 15 nodes. To enforce the distributed placement of nodes, each virtual node may only be mapped to a quarter of the substrate nodes. As all studied algorithms scale alike in the number of requests, we fix the number of requests per scenario to be 30. For each of the requests, the profit is set proportionally to the minimal node and edge resource usage. We consider the following parameters to draw resource demands given uniform substrate node and edge capacities (cf.~\cite{rostTON20randround}).

\noindent\textbf{Node resource factor} (\nrf): We consider values in $\{0.3, 0.8\}$, implying (averaged) node utilizations of 30\% and 80\%, respectively.

\noindent\textbf{Edge resource factor} (\erf): We consider values in $\{0.3, 0.8\}$, such that the cumulative bandwidth demand of all requests equals all available bandwidth capacities divided by \erf. Accordingly, edge resources are generally scarce.

\begin{figure*}[t!]
	\begin{minipage}{0.485\textwidth}
		\begin{subfigure}{0.495\textwidth}
			\includegraphics[width=\linewidth]{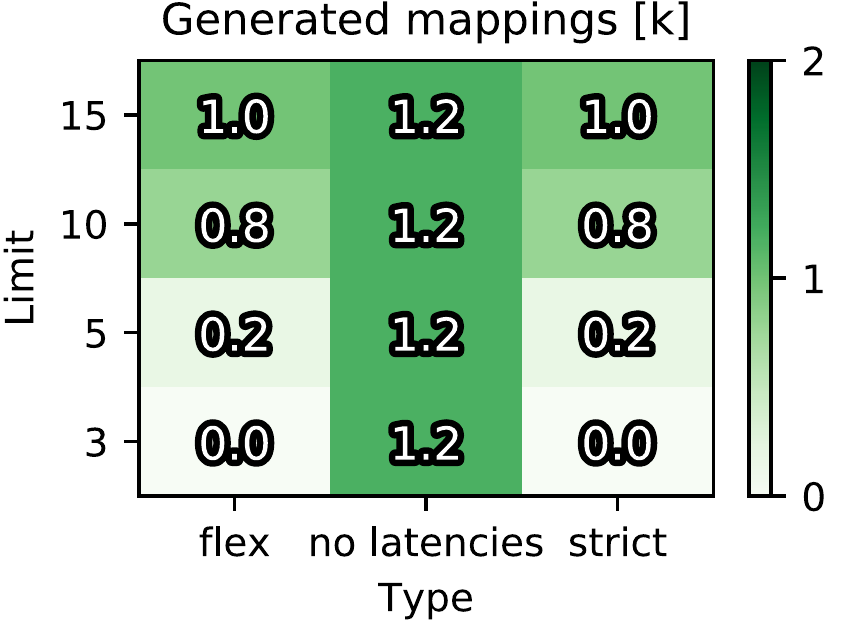}
			\caption{}
			\label{fig:sparserequests.mappings}
		\end{subfigure}
		\hfill
		\begin{subfigure}{0.495\textwidth}
			\includegraphics[width=\linewidth]{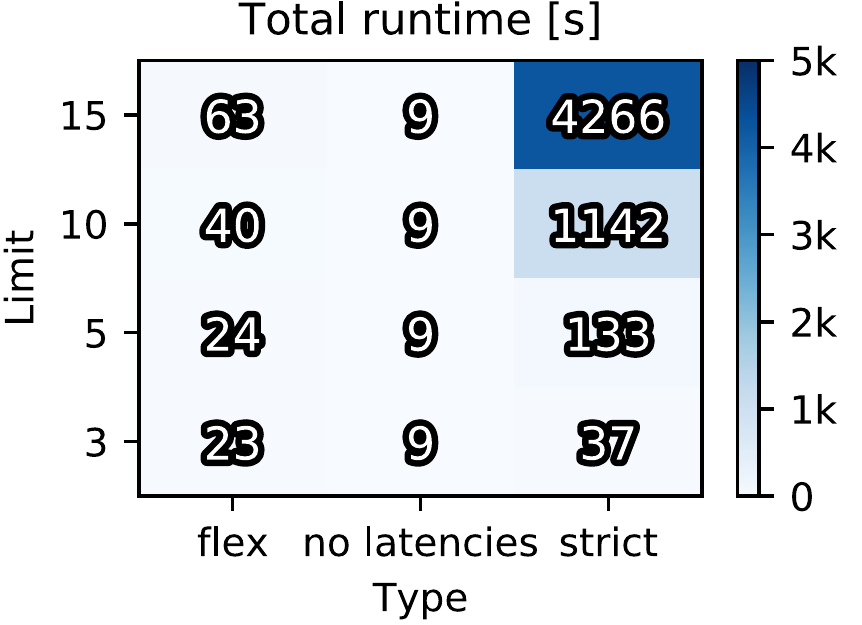} 
			\caption{}
			\label{fig:sparserequests.runtime}
		\end{subfigure}
		\caption{Number of generated mappings and total runtime for expected low node and high edge utilizations. Parameters: \erf: 0.3, \nrf: 0.3, $\epsilon$: 0.1, averaged over all topologies.}
		\label{fig:sparserequests}
	\end{minipage}	
	\hfill
	\begin{minipage}{0.485\textwidth}
		\begin{subfigure}{0.495\textwidth}
			\includegraphics[width=\linewidth]{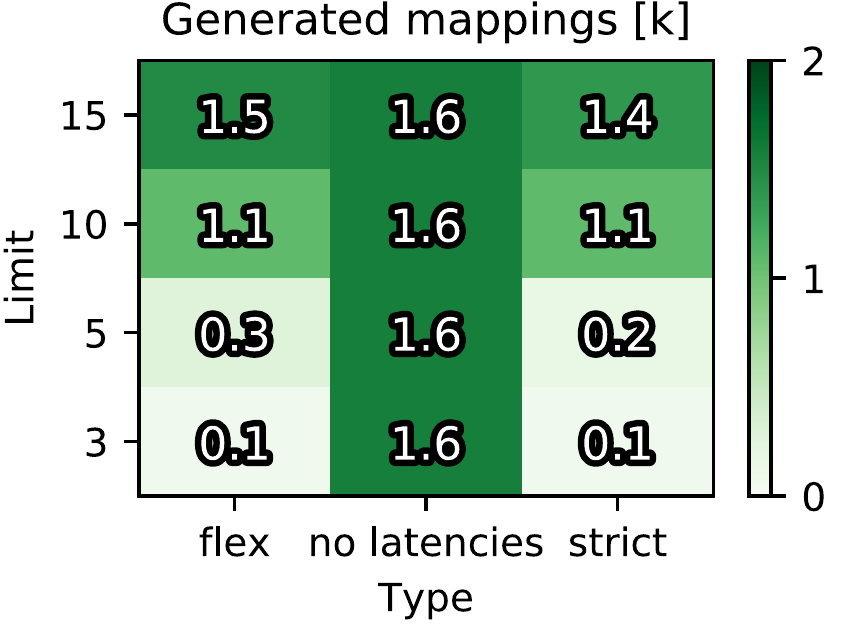}
			\caption{}
			\label{fig:denserequests.mappings}
		\end{subfigure}
		\hfill
		\begin{subfigure}{0.495\textwidth}
			\includegraphics[width=\linewidth]{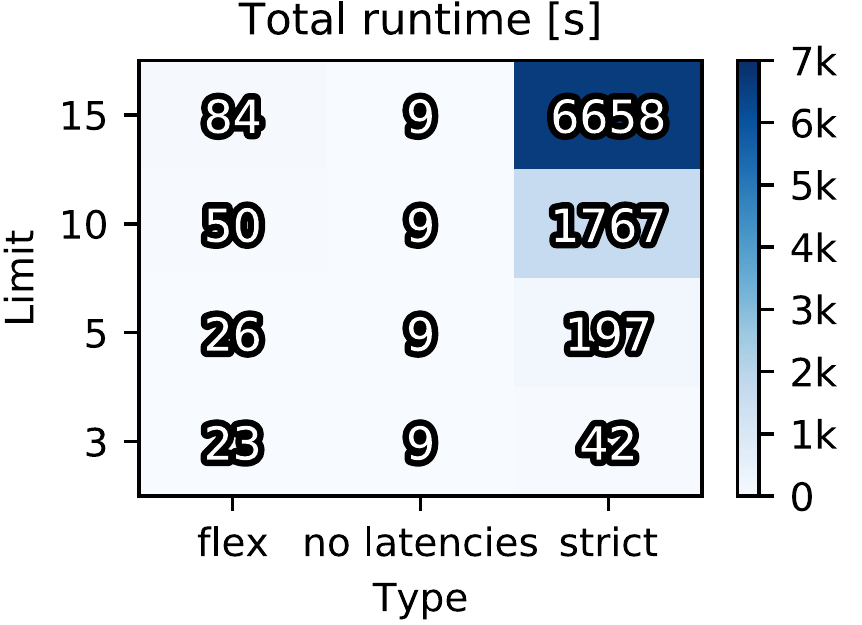} 
			\caption{}
			\label{fig:denserequests.runtime}
		\end{subfigure}
		\caption{Number of generated mappings and total runtime for expected high node and lower edge utilizations. Parameters: \erf: 0.8, \nrf: 0.8, $\epsilon$: 0.1, averaged over all topologies.}
		\label{fig:denserequests}
	\end{minipage}
	
	\vspace{12pt}
	
	\begin{minipage}{0.485\textwidth}
		\begin{subfigure}{0.495\textwidth}
			\includegraphics[width=\linewidth]{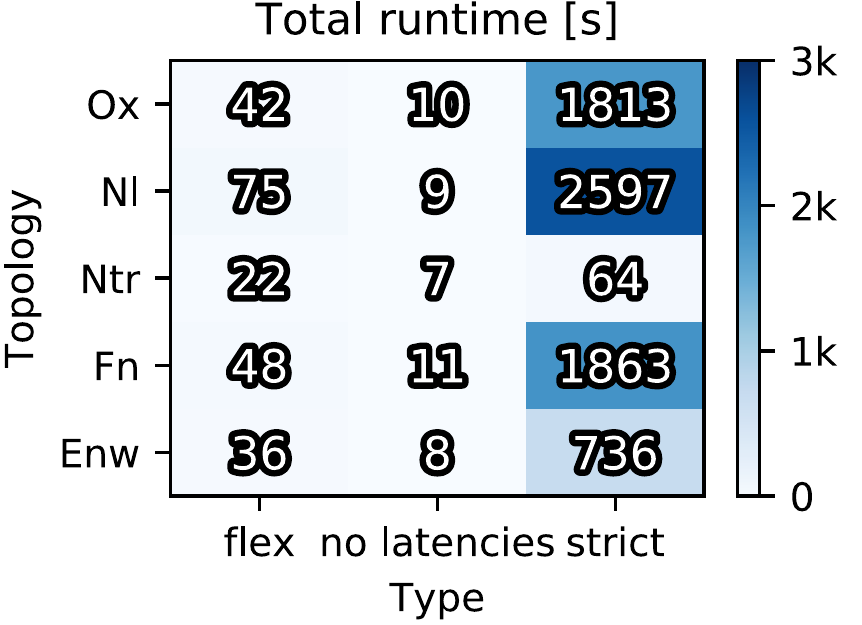} 
			\caption{$\epsilon = 0.1$}
			\label{fig:substratecomp.1}
		\end{subfigure}%
		\hfill
		\begin{subfigure}{0.495\textwidth}
			\includegraphics[width=\linewidth]{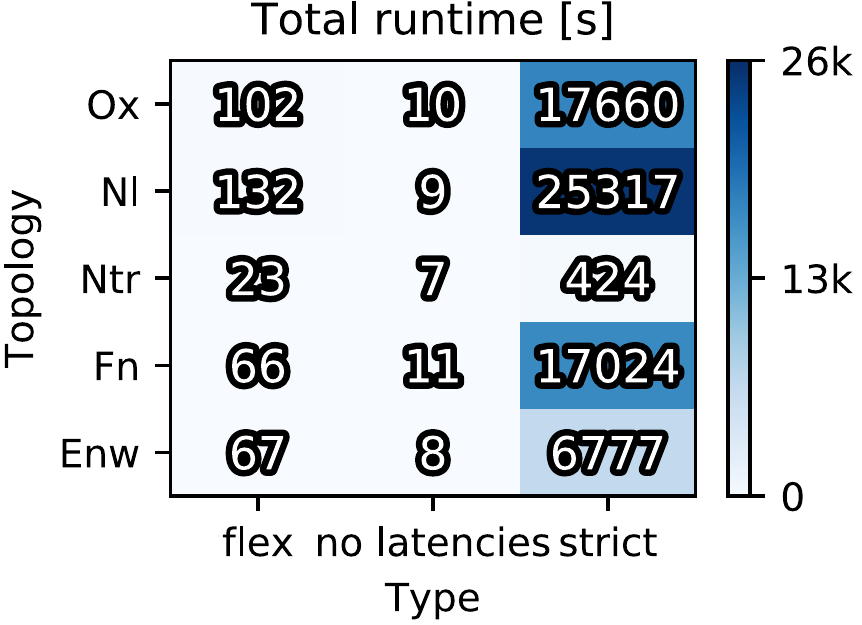} 
			\caption{$\epsilon = 0.02$}
			\label{fig:substratecomp.02}
		\end{subfigure}
		\caption{Total runtime per substrate network, split by different values of $\epsilon$. Parameters: limit: 10, averaged for \erf and \nrf.}
		\label{fig:substratecomp}
	\end{minipage}
	\hfill
	\begin{minipage}{0.485\textwidth}
		\begin{subfigure}{0.495\textwidth}
			\includegraphics[width=\linewidth]{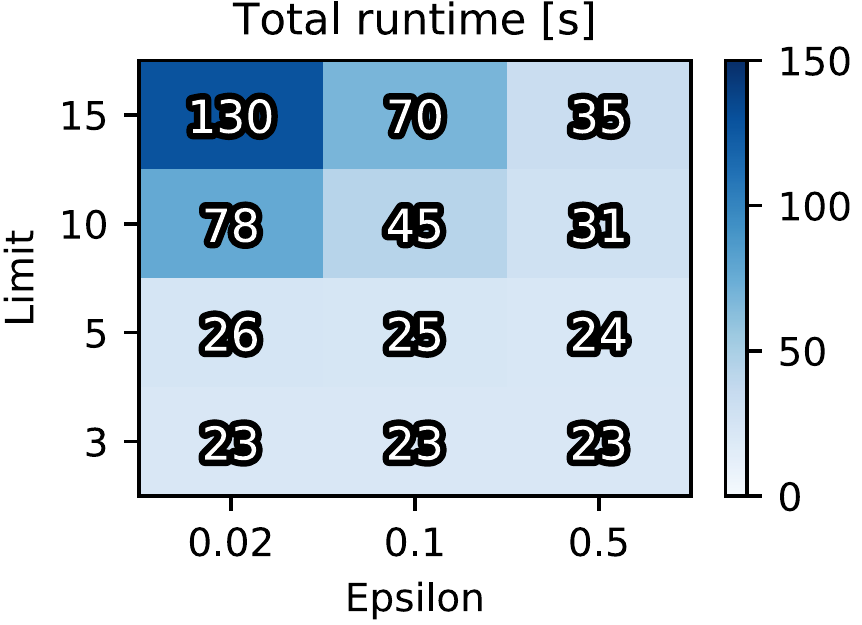} 
			\caption{\AFlex}
			\label{fig:eps.flex}
		\end{subfigure}%
		\hfill
		\begin{subfigure}{0.495\textwidth}
			\includegraphics[width=\linewidth]{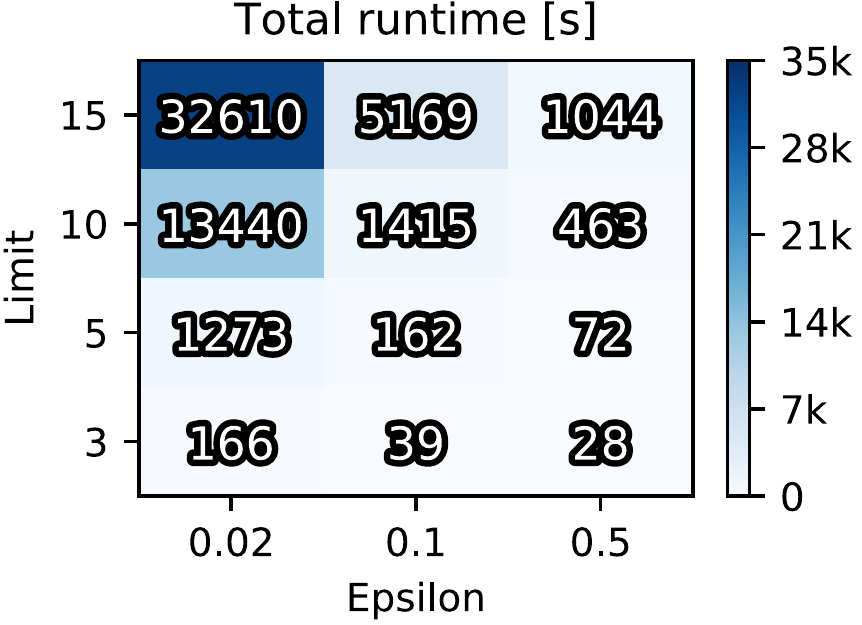} 
			\caption{\AStrict}
			\label{fig:eps.strict}
		\end{subfigure}
		\caption{Total runtime for different values of limit and $\epsilon$, split by algorithm type. The results are averaged for \erf, \nrf and over all topologies.}
		\label{fig:eps}
	\end{minipage}
\end{figure*}

For this paper, the execution parameters are extended by the following values.

\noindent\textbf{Latency approximation type}: Specifies which algorithm is used for calculating valid paths. May either be \AStrict, \AFlex, or the baseline (disregarding latencies).

\noindent\textbf{Latency approximation factor $\epsilon$:} We consider values in  $\{0.5, 0.1, 0.02\}$ as approximation factors for the RSP approximation algorithms by Lorenz and Raz and by Goel et al.

\noindent\textbf{Latency limit scaling factor:} We consider scaling values of $\{3, 5, 10, 15\}$ to set the latency bounds of the request graphs. Specifically, the scaling factor is multiplied by the substrate's average edge latency $\phi(G_S)$ to obtain the imposed latency limit. Accordingly, a scaling factor of $5$ limits the realization of request edges to at most 5 average-latency substrate edges.

We generate scenarios according to the 20 different parameter combinations (topology, edge and node resource factor). For each parameter combination, we generate 9 unique instances. Each scenario instance is then executed once using the baseline algorithm (neglecting latencies) and 12 times (each combination of latency approximation factor, limit scaling factor) for \AFlex and \AStrict.

Importantly, we always report on \emph{fractional} VNEP solutions obtained by using the different algorithms. We do so, as the rounding step necessary to obtain integral solutions is \emph{the same} for all algorithms and would hence only introduce additional randomness.

\subsection{Runtime Comparison} 

We now present the results of the runtime evaluation of the three algorithms.
Following the theoretical analysis the runtime is expected to increase with increasing substrate size and with decreasing approximation factor $\epsilon$. In the following we analyze the practical sensitivity of the three algorithms to changes to these parameters.
 \begin{figure*}[t]
	\begin{subfigure}{0.485\textwidth}                             
		\centering
	\includegraphics[width=0.9\linewidth]{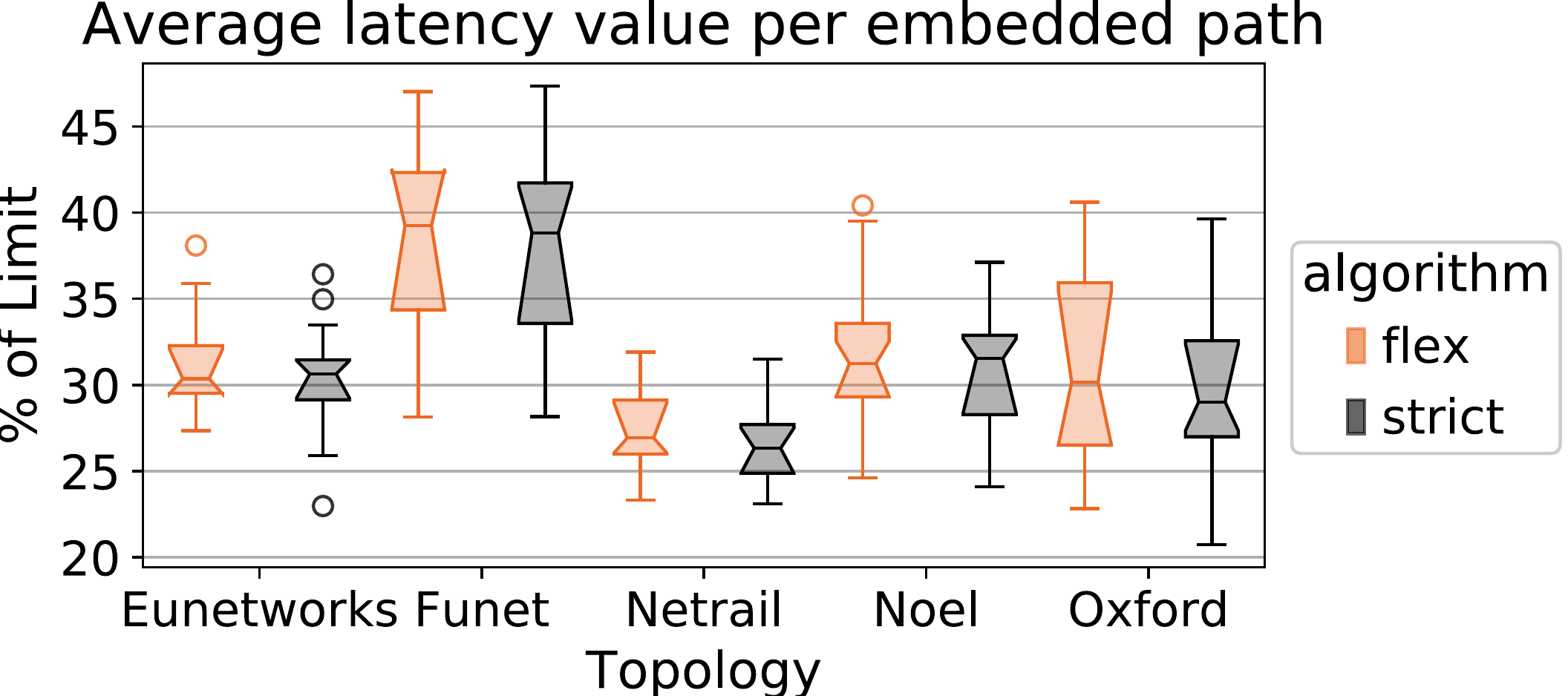} 
\end{subfigure}
	\hfill
	\begin{subfigure}{0.485\textwidth}
				\centering
	\includegraphics[width=0.9\linewidth]{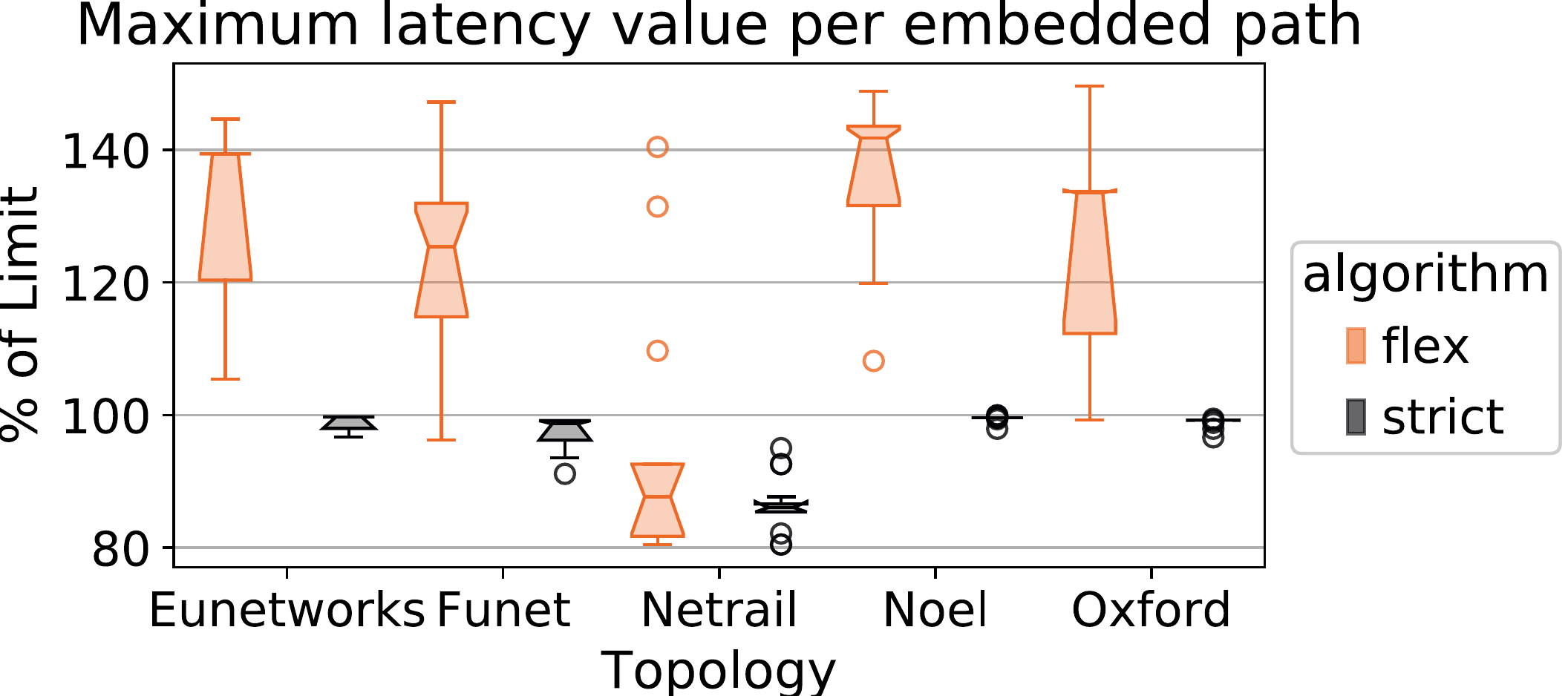} 
\end{subfigure}%

	\caption{Boxplots of the average (left) and maximum (right) latency value of each embedded request edge for $\epsilon=0.5$  and latency limit factor of 10.}
	\label{fig:boxplot:flex:latency:eps}
\end{figure*}

We first consider the runtime as a function of the algorithm and the enforced latency limit. Figure~\ref{fig:sparserequests} shows the results for lower expected node utilizations and higher edge utilizations are shown while Figure~\ref{fig:denserequests} shows the results for high expected node demands and lower edge utilizations. The baseline does not depend on the limit value. As a result the values in the middle column are always identical. 
Regarding the number of generated mappings, we observe that for \AFlex and \AStrict small limit values significantly restrict the amount of generated mappings while for large latency values the number of generated mappings approaches the baseline's value (cf. Figures~\ref{fig:sparserequests.mappings} and~\ref{fig:denserequests.mappings}). This in turn validates our experimental design: a limit of 3 is indeed very restrictive while a limit of 15 allows for all but few valid mappings found by the baseline. 
 Figures~\ref{fig:sparserequests.runtime}~and~\ref{fig:denserequests.runtime} show the runtime of the algorithms. While the runtime clearly is dependent on the number of returned mappings, the \AStrict algorithm is computationally much more expensive than \AFlex. This can be seen most dramatically for loose latency limits and relaxed edge demands in Figure~\ref{fig:denserequests.runtime}: the average runtime of \AFlex lies beneath 1.5 minutes while the runtime of \AStrict approaches roughly 1.9 hours.  
 Further investigating the runtime discrepancy,  Figure~\ref{fig:substratecomp} shows the runtime as a function of the substrate topology and the approximation guarantee. Firstly, we observe that the runtime of \AStrict lies significantly above both other algorithms and increases drastically with the substrate size (cf. Table~\ref{tab:zoo}). Also, the approximation factor $\epsilon$ has a much more dramatic impact on the runtime for \AStrict as the runtime increases roughly by a factor of 8.9 on average. These observations are in accordance with the theoretical runtime bounds from Section~\ref{section:application}, as the runtime of \AStrict is larger than the one of \AFlex by a factor of at least $\Omega(poly(n))$ (cf.~Theorems~\ref{theorem:strict} and~\ref{theorem:contribution:flex}).  Interestingly, both latency-algorithms take the longest on the topology Noel, even though it is neither the largest substrate in terms of nodes nor in terms of edges.  Lastly, Figure~\ref{fig:eps} depicts the runtime of the respective algorithms as a function of the approximation guarantee $\epsilon$ and the latency limit. Clearly, \AFlex provides a much better scalability both in terms of the latency limit and the approximation guarantee: for the highest latency limit and the best approximation factor \AStrict's runtime averages to about 9 hours while \AFlex only takes less than 2.2 minutes.  Importantly, In Figure~\ref{fig:eps.flex} we can observe another favorable quality of the \AFlex algorithm. Specifically, for small limit values (3 and 5) when very few mappings are generated, \AFlex shows no change in runtime for different values of $\epsilon$.

 \subsection{Latency Comparison}
In the following we will study how the algorithm choice influences the per edge latency. Figure~\ref{fig:boxplot:flex:latency:eps} illustrates the average as well as the maximum latency of the solutions on the different topologies. Recalling that the \AStrict algorithm may not exceed latency bounds while \AFlex may do so up to a factor of $1+\epsilon$ we can observe that the average latency of the \AFlex algorithm often times lies slightly above the one of \AStrict. Furthermore, and very importantly, for the medium latency limit of $10$ the average edge latency lies strictly below 50\% of the limit for all topologies and algorithms. However, there are cases in which both algorithms make full use of the maximal latency limit. In fact, the maximum edge latency of the \AStrict algorithm often comes very close to the imposed limit factor while the \AFlex algorithm rarely reaches its upper bound of $1+\epsilon$ times the original limit. Notably, these results empirically validate the correctness of our implementation.
Regarding the implications of these results, we observe that the average edge latency often times lies strongly below the imposed limit. This may hold true especially since compact embeddings, i.e., ones that use the least bandwidth resources, allow for embedding more requests, thereby increasing the profit. 
 Accordingly, it may be reasonable to at first chose a larger value for $\epsilon$ to significantly reduce the runtime while also keeping the average latency values low, with the option to refine the choice of $\epsilon$ whenever the latency-limit violation is too large.

\subsection{Profit Comparison}

We will now shortly analyze the \emph{quality} of the solutions produced by \AFlex, \AStrict, and the baseline algorithm. As performance measure, we employ the achieved profit of the computed fractional solutions. 

Figure~\ref{fig:boxplot} shows the averaged profit for the various topologies, the largest approximation factor $\epsilon$ and a medium latency limit factor of $10$. One can first observe that the profits of the \AFlex and the \AStrict algorithm are very similar while the baseline's profit regularly slightly exceeds the latency limited algorithms. While this is to be expected and our theoretical observation implies that the profit of \AFlex should be slightly above \AStrict in the most cases, we also observe some rare cases in which the profit of the latency limited algorithms exceeds the profit of the baseline. We believe this to be due to numerical instabilities when solving the underlying linear programs. 

\begin{figure}[t]
	\centering
	\includegraphics[width=0.9\linewidth]{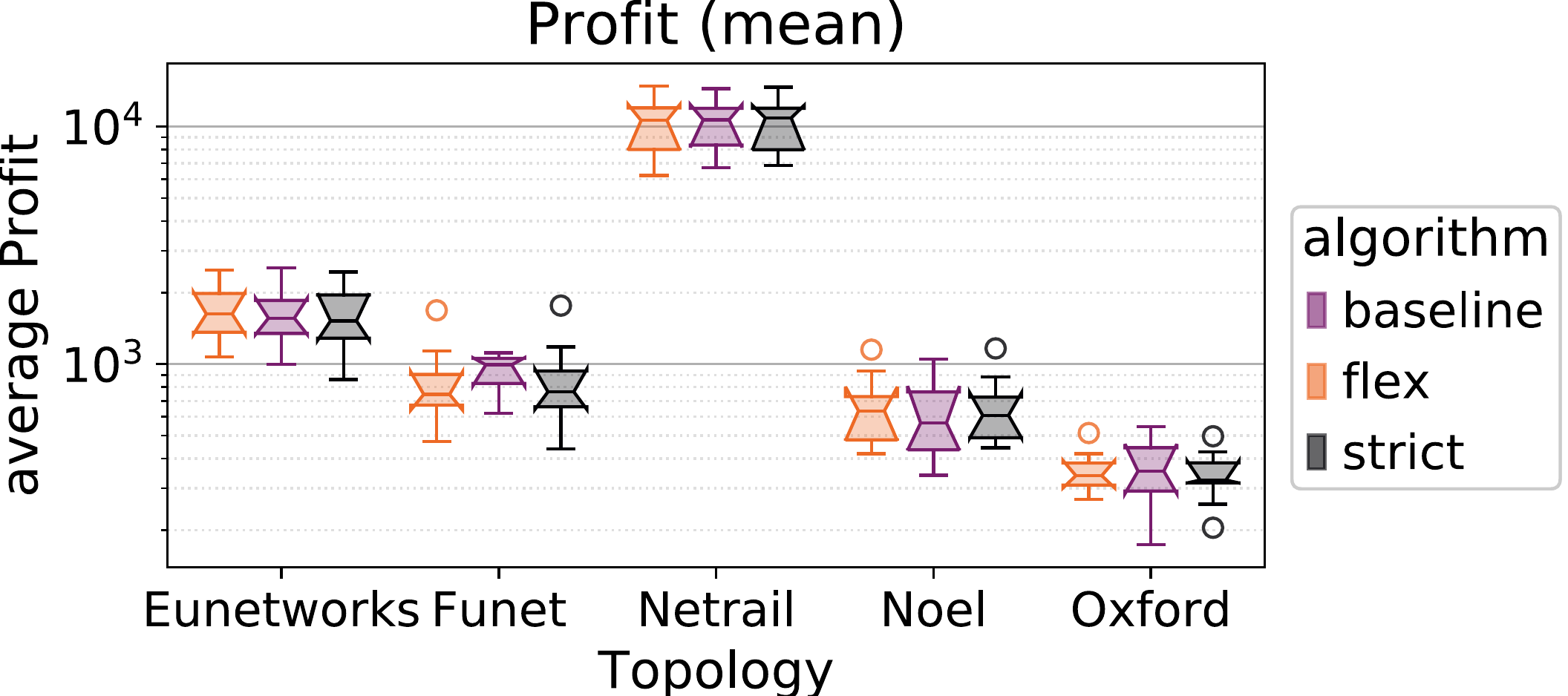} 
	\caption{Boxplot of the average achieved profits per topology. Parameters: $\epsilon: 0.5$, limit: 10.}
	\label{fig:boxplot}
\end{figure}

\begin{table}[b!]
	\centering
	\begin{tabular}{ c|c|c } 
		Substrate Network & Nodes & Edges \\ %
		\hline
		Geant (2012)  & 40 & 122 \\  %
		Iris  & 51 & 128 \\ %
		UsSignal  & 63 & 158 
	\end{tabular}
	\caption{Networks for scalability study of \AFlex }
	\label{tab:zoo-scalability}
\end{table}

\begin{figure*}[t]
\begin{minipage}[b]{0.485\textwidth}
	\centering
	\includegraphics[width=0.9\linewidth]{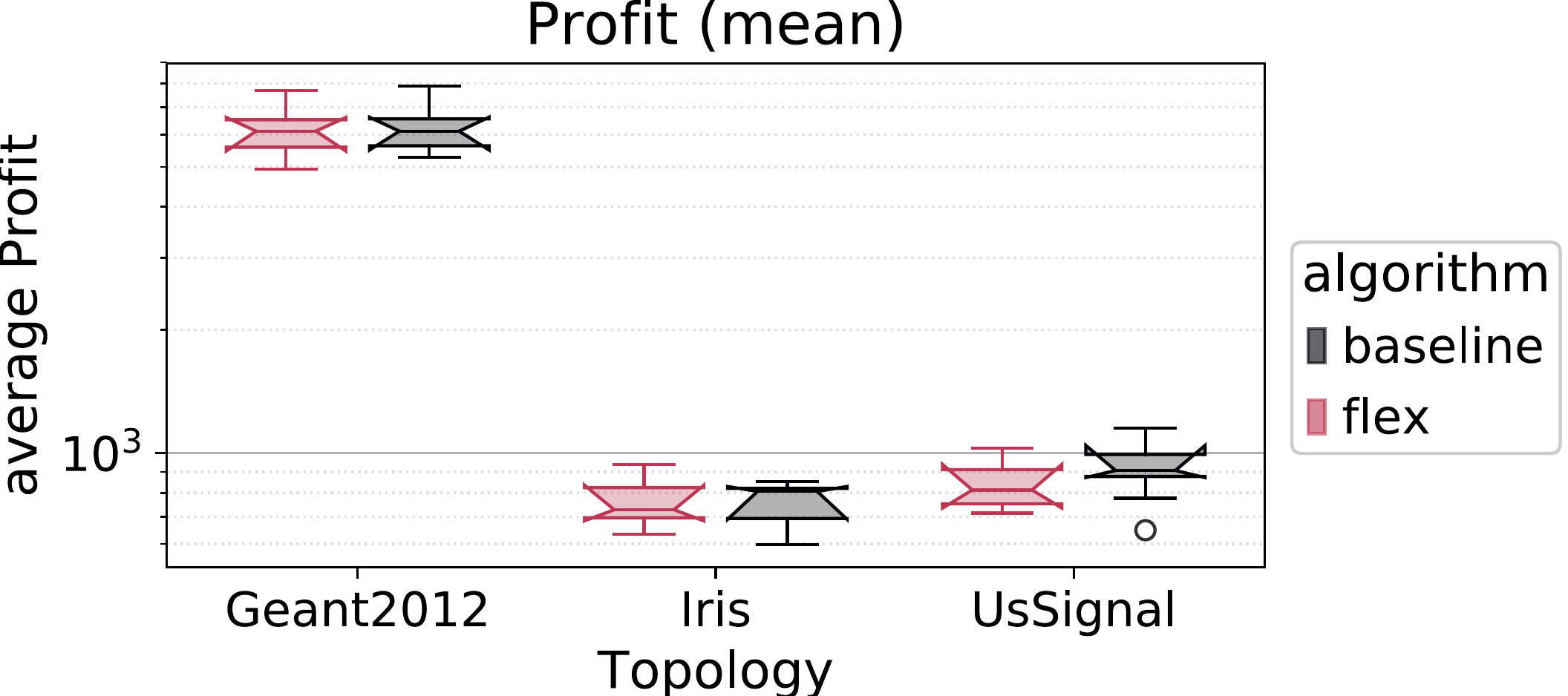} 
	\caption{Plot of the average achieved profit per topology of the scalability study for $\epsilon=0.5$ and the limit 10.}
	\label{fig:boxplot:flex:profit}
\end{minipage}	
\hfill
\begin{minipage}[b]{0.485\textwidth}
	\begin{subfigure}{0.5\columnwidth}
		\includegraphics[width=\linewidth]{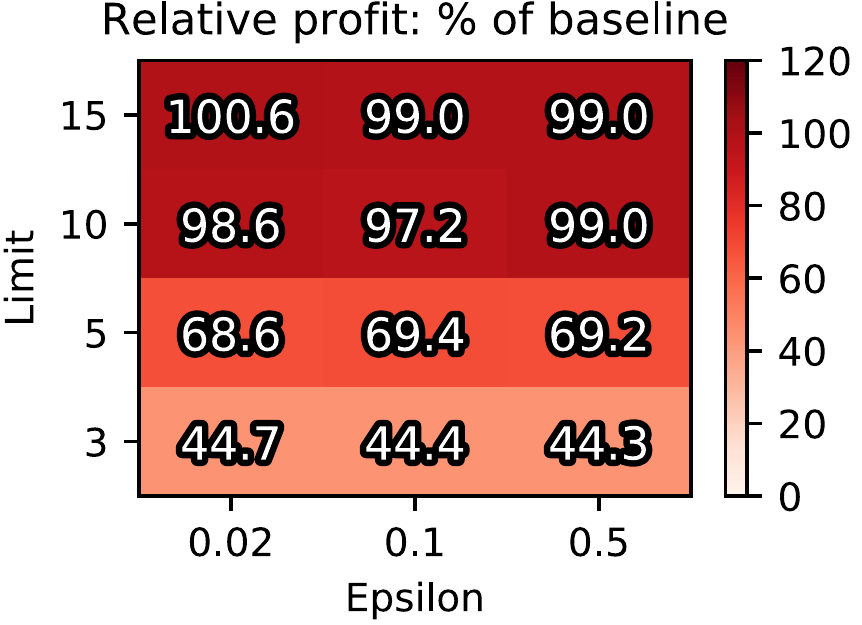}
	\end{subfigure}%
	\begin{subfigure}{0.5\columnwidth}                              \includegraphics[width=\linewidth]{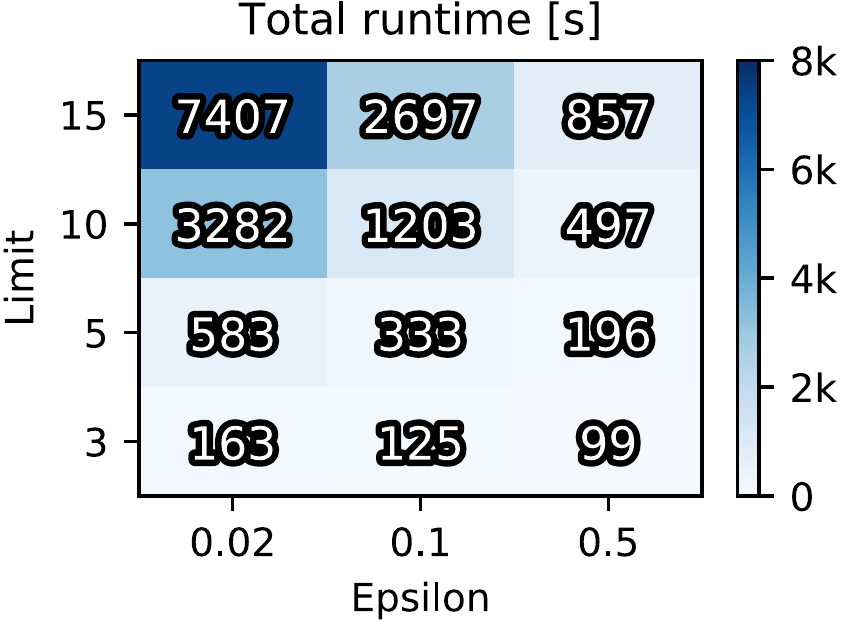}
	\end{subfigure}
	\caption{Relative maximum achieved profit and total runtime of \AFlex of the scalability study.}
	\label{fig:performance.eps}
\end{minipage}

\end{figure*}

\subsection{Scalability of \AFlex and Discussion}
\label{sec:evaluation:aflex-performance}
To emphasize how well \AFlex performs in direct comparison with the no-latency-baseline we also conducted an explorative study on larger substrate networks obtained from the Topology Zoo (see Table~\ref{tab:zoo-scalability}). Given the substrate sizes of at least 40 nodes and 122 edges, all of these substrates are far too large for running \AStrict in a reasonable amount of time. For these additional experiments, we adapt the scenario parameters as follows. We consider 50 requests and employ the same request generation procedure as before using \erf and \nrf values of 0.5. For each topology, we consider 10 instances created at random.

\AFlex produces solutions that yield profits close to the baseline for a latency limit of 10 in all substrates, see Figure~\ref{fig:boxplot:flex:profit}. In Figure~\ref{fig:performance.eps} we observe that for small $\epsilon$ and large limit values the average runtime of \AFlex slightly exceeds two hours, but stays below one hour in the other cases. Since the $\epsilon$ only scales the relaxation of the latency constraint, it has little influence on the achieved profits.

\section{Conclusion}\label{section:conclusion}

This paper presented a novel approximation algorithm
for the embedding of virtual networks which accounts
for latency constraints. Our algorithm is significantly
faster than state-of-the-art algorithms, as we have also
shown empirically. 
We believe that the combination of formal approximation
guarantees and low runtime makes our algorithm particularly interesting
in practice, as it allows to include latency constraints with little overhead. 
To ensure reproducibility and facilitate future
research in this area we have made the source code
of our algorithms and our experiments publicly available.

\section*{Acknowledgments.}
This project received funding from the 
European Research Council (ERC) under grant agreement 864228 (AdjustNet), Horizon 2020, 2020-2025.

{\balance
\bibliographystyle{abbrv}
\bibliography{bibliography}
}

\end{document}